\documentclass[journal]{IEEEtran}

\usepackage{xcolor,soul,framed} 

\colorlet{shadecolor}{yellow}
\usepackage[pdftex]{graphicx}
\graphicspath{{../pdf/}{../jpeg/}}
\DeclareGraphicsExtensions{.pdf,.jpeg,.png}

\usepackage[cmex10]{amsmath}
\usepackage{array}
\usepackage{mdwmath}
\usepackage{mdwtab}
\usepackage{eqparbox}
\usepackage{url}
\usepackage{amssymb}
\usepackage{booktabs}
\usepackage{caption}
\usepackage{subcaption}
\usepackage[numbers,square]{natbib}
\usepackage{amsmath}
\usepackage{enumerate}

\newtheorem{theorem}{Theorem}

\newtheorem{remark}{Remark}

\newtheorem{lem}{Lemma}
\newtheorem{definition}{Definition}
\newtheorem{assum}{Assumption}
\begin{document}
\bstctlcite{IEEEexample:BSTcontrol}
    \title{Payoff distribution in robust coalitional games on time-varying networks}
  \author{Aitazaz Ali Raja and
     Sergio Grammatico
\thanks{Aitazaz Ali Raja and Sergio Grammatico are with Delft Center for Systems and Control, TU Delft, The Netherlands. (e-mail addresses: a.a.raja@tudelft.nl; s.grammatico@tudelft.nl).}
  \thanks{This work was partially supported by NWO under research project P2P-TALES (grant n. 647.003.003) and the ERC under research project COSMOS, (802348).}
}  
\maketitle
\begin{abstract}
In this paper, we consider a sequence of transferable utility (TU) coalitional games where the coalitional values are unknown but vary within certain bounds. As a solution to the resulting family of games, we formalize the notion of ``robust core". Our main contribution is to design two distributed algorithms, namely, distributed payoff allocation and distributed bargaining, that converge to a consensual payoff distribution in the robust core. We adopt an operator-theoretic perspective to show convergence of both algorithms executed on time-varying communication networks. An energy storage optimization application motivates our framework for ``robust coalitional games".
\end{abstract}


%
\IEEEpeerreviewmaketitle


\section{Introduction}\label{sec: intro}
\IEEEPARstart{C}{oalitional} game theory provides a framework to study the behavior of selfish and rational agents when they cooperate effectively. This willingness to cooperate arise from the aspiration of gaining a higher return, compared to that for behaving as individuals \cite{myerson2013game}. \\
Specifically, a transferable utility (TU) coalitional game consists of a set of agents and a value/characteristic function that provides the \textit{value} of each of the possible coalitions \cite{myerson2013game}. Multi-agent decision problems modelled by TU coalitional games arise in many application areas, such as demand-side energy management \cite{han2018incentivizing} and cooperation between microgrids \cite{Saad2011}, in various areas of communication networks \cite{Saad2009} and as the foundation of coalitional control \cite{fele2017coalitional}.\\ 
One key problem studied by coalitional game theory is the distribution of the value generated by cooperation. Along this research direction, several solution concepts have been proposed with special attention to criteria like stability and fairness. In payoff distribution, stability means that none of the agents has an incentive to defect the coalition. Perhaps the most studied solution concepts in coalitional games that ensures the \textit{stability} of a payoff is the core. The second criterion, i.e., \textit{fairness} means that the payoff for an agent should reflect its contribution to or impact in the game. A seminal work on  the axiomatic characterization of fairness is that of Shapley \cite{shapley1953value}, where the unique value satisfying the fairness axioms is in fact known as the Shapley value which depends on the marginal contribution of each agent. The later depicts the impact each agent has on the collective value of the coalition. Other related solution concepts are also proposed in the literature, e.g. the  Nucleolus and the Kernel \cite{maschler1992bargaining}.\\
In this paper, we consider the problem of finding a payoff distribution that encourages cooperation, i.e., belongs to the core \cite{raja2020}. Now, to evaluate such a payoff, the value of each possible coalition is required, which seems implausible in many practical applications, mainly because an agent cannot be certain about the values that collaborations may generate. However, one can assume that an agent does hold a belief about the value of some possible collaborations via informed estimation or mere experience. In practice, this brings uncertainty to the coalitional values and, consequently, to the core set. It follows that one should consider solutions that are robust to uncertainty on the coalitional values. In this paper, we do that via the notion of \textit{robust }core.\\
The robustness aspect in coalitional games falls into the framework of dynamic TU coalitional games, which has been studied in the literature. Among others, the authors in \cite{filar2000dynamic} analyzed the time consistency of the Shapley value and the core under the temporal evolution of the game. Then, the authors in \cite{lehrer2013core} characterized three versions of core allocations for a dynamic game where the worth of the coalitions varies over time according to the previous allocations. In both papers, the coalitional values at the current time are determined endogenously and depend on previous events. {In \cite{kranich2005core}, the authors consider a finite sequence of exogenously defined coalitional games, where the agents receive a payoff at each stage of the sequence and consequently, the final utility of an agent depends on the whole stream of payoffs. }\\
Robust coalitional games are the subclass of dynamic TU coalitional games where the coalitional values are unknown and exogenous. 
In \cite{bauso2009robust}, Bauso and Timmer characterized robust allocation rules for the dynamic coalitional game where the average value of each coalition is known with certainty, while at each instant, the coalitional value fluctuates within a bounded polyhedron. The static version of their setup, called cooperative interval games, is presented by the authors in \cite{alparslan2009cooperation},  where the coalitional values are considered yet to be uncertain within some bounded intervals. In their setup, they have introduced the interval solutions, which assign a closed real interval as a payoff to each agent instead of a single real value. In \cite{Nedic2013}, Nedich and Bauso have presented a distributed bargaining algorithm for finding a solution in the core under the framework of robust games and dynamic average games. Inspired by the motivation of cooperative interval games and the setup in \cite{Nedic2013}, in this paper, we consider the value generated by each coalition to vary within certain bounds. \\
\textit{Motivational example}: Let us consider the energy optimization application inspired by \cite{han2018constructing} which justifies a dynamic robust coalitional game model. Consider a group of $N$ prosumers, each of whom owns a renewable energy source (RES) and energy storage (ES). Together they form an energy coalition $\mathcal{I}$ where the participating agents operate their ES systems collectively to minimize their total energy cost. When the energy coalition has an excess of energy, they can store it in an ES for later utilization and any additional energy can be sold to a retailer. The retailer buys energy and remunerates, a few hours ahead of the delivery time. The coalition considers the corresponding remuneration for optimizing their ES operation and consequently minimizing the associated cost function.  \\
Now, the cost saving as a result of the collaborative operation should be distributed in such a way that each prosumer is satisfied by its share, and hence the coalition remains intact. To achieve this, the agents assert their position by presenting the estimated cost saving of possible energy sub-coalitions, $S \subseteq \mathcal{I}$, which they could have been part of and use them to define acceptable payoffs, namely payoffs in the core. Since there is uncertainty in the RES generation, the cost savings, $v(S)$, of each sub-coalition $S \subseteq \mathcal{I}$ is uncertain. How the agents share this saving under such uncertainty is a part of the solution generated by an iterative payoff distribution methods. \\
Let \(b_{i}^t\) represent the charge or discharge of energy by the ES of prosumer $i$ at time $t$. Further, denote the net energy demand of prosumer $i$ by \(q_{i}^{t}\) and let \(p_{\text{s}}^t\) and  \(p_{\text{b}}^t\) be an electricity sell price and buy price at time $t$, respectively. { Let  $\mathrm{proj}_{\geq 0}(x) (\mathrm{proj}_{\leq 0}(x))$  denote the projection onto the non-negative (non-positive) orthant.} Then, the energy cost function of any energy sub-coalition $S \subseteq \mathcal{I}$ for a time period of length $K$ is given as:
\begin{equation*}
\begin{array}{ll}
F_S(\boldsymbol{b}):= \sum_{t =1}^K \bigg\{ p_{\text{b}}^t \bigg(&\sum_{i \in S}\mathrm{proj}_{\geq 0 }(q_{i}^{t}+b_{i}^t)\bigg)\\
&+ p_{\text{s}}^t \bigg(\sum_{i \in S}\mathrm{proj}_{\leq 0 }(q_{i}^{t}+b_{i}^t)\bigg)\bigg\},
\end{array}
\end{equation*}
where $\boldsymbol{b} \in \mathbb{R}^{NK}$ contains the ES charge and discharge profiles of all the $N$ agents over the $K$ time steps. \\
For a given coalition $S$, the coalitional energy cost for the time period of length $K$ is defined as:
\begin{equation}\label{eq: coalitional cost}
c (S) := \min _{\mathbf{b}} F_S(\mathbf{b}),
\end{equation}
and the  cost saving during this period, \(v (S)\), as the difference between the sum of the costs of the coalitions of the individual agents in $S$ and the cost of the coalition itself, namely,
\begin{equation}\label{eq: coalitional value}
\textstyle  v (S):=\sum_{i \in S} \{c_{i}\}-c (S).
\end{equation}
Note that, the cost $c(S)$ is unknown but bounded, from above when each agent $i \in S$ has RES generation equal to the installed capacity, which gives minimum value of net energy consumption $q_i^{\text{min}}$ for the whole period $K$, and from below when there is no generation, hence $q_i^\text{max}$. Due to these bounds, the cost saving \(v (S)\) is also bounded. Let $\underline{c}(S)$ be the coalitional cost corresponding to $q_i^{\text{max}}$, and let $\overline{c}_i$ be the individual cost corresponding to $q_i^{\text{min}}$, $i \in S$ then:
$$
 \textstyle v (S) \leq \sum_{i \in S} \{\overline{c}_i\} - \underline{c}(S).
$$
The uniform upper bound on the coalitional values and the fixed value of grand coalition, for a period of length $K$, allows us to consider the setup of robust games presented in \cite{Nedic2013}. {We refer to \cite{baeyens2013coalitional}, \cite{chakraborty2018sharing}, \cite{feng2019coalitional} for other engineering problems that can be modeled as robust coalitional games.} \\
\textit{Contribution}: We propose two payoff distribution algorithms within the framework of robust coalitional games where the values of the coalitions are time-varying:
\begin{itemize}
\item We {formalize} the notion of \textit{robust} core, a set of payoffs that stabilizes a grand coalition under variations in the coalitional values (Section \ref{sec:Mathematical background });
\item We develop a distributed payoff allocation algorithm where agents communicate only locally, i.e., with their neighbors, over a time-varying and repeatedly-connected communication network. We show that the proposed algorithm converges to a common payoff allocation in the \textit{robust} core (Section \ref{sec: payoff allocation});
\item {We generalize the distributed bargaining protocol in \cite{Nedic2013}}  and prove its convergence to a mutually agreed payoff in the \textit{robust} core. We assume similar communication requirements for the bargaining protocol as for the allocation process; but less information on the game is available to the agents (Section \ref{sec: bargaining});
\item {We introduce some tools from operator theory (\textit{paracontraction}, \textit{nonexpansive} operators and Krasnoselskii-Mann fixed-point iterations) to the domain of coalitional games which allows us to generalize existing results and in turn to propose faster algorithms. This approach represents a new general analysis framework for coalitional games.} 
\end{itemize}

\textit{Notation}: $\mathbb{R}$ and $\mathbb{N}$ denote the set of real and natural numbers, respectively. Given a mapping $M: \mathbb{R}^n \rightarrow \mathbb{R}^n, \mathrm{fix}(M):= \{x \in \mathbb{R}^n \mid x = M(x)\} $ denotes the set of its fixed points. $\text{Id}$ denotes the identity operator. For a closed set \(C \subseteq \mathbb{R}^{n},\) the mapping $\mathrm{proj}_C$: \(\mathbb{R}^{n} \rightarrow C\) denotes the projection onto \(C,\) i.e., \(\operatorname{proj}_C(x)=\) \(\arg \min _{y \in C}\|y-x\| .\) An over-projection operator is denoted by $\mathrm{overproj}_C : = 2\mathrm{proj}_C - \text{Id} $. For a set $S$ the power set is denoted by $2^S$. \(A \otimes B\) denotes the Kronecker product between the matrices \(A\) and \(B .\) $I_N$ denotes an identity matrix of dimension $N \times N$. {For $x_{1}, \ldots, x_{N} \in \mathbb{R}^{n},$ $\mathrm{col}(\left(x_{i}\right)_{i \in(1, \ldots, N)}):=\left[x_{1}^{\top}, \ldots, x_{N}^{\top}\right]^{\top}.$} For a norm \(\|\cdot\|_{p}\) on \(\mathbb{R}^{n}\)
and a norm \(\|\cdot\|_{q}\) on \(\mathbb{R}^{m},\) the mixed vector norm \(\|\cdot\|_{p, q}\) on \(\mathbb{R}^{m n}\) is defined as $\|x\|_{p, q} = \|\mathrm{col}(\|x_{1}\|_{p}, \cdots, \|x_{m}\|_{p} )\|_{q} $. $\mathrm{dist}(x,C)$ denotes the distance of $x$ from a closed set \(C \subseteq \mathbb{R}^{n},\) i.e., $\mathrm{dist}(x,C):= \mathrm{inf}_{y \in C} \|y-x\|$. {For a closed set \(C \subseteq \mathbb{R}^{n}\) and $N \in \mathbb{N}, C^N:=\prod_{i=1}^{N} C_{i}$}.\\
\textit{Operator-theoretic definitions}:
A mapping $T : \mathbb{R}^n \rightarrow \mathbb{R}^n$ is nonexpansive, if $ \|T(x) - T(y)\| \leq \|x - y\|,$
for all $x,y \in \mathbb{R}^n$.
A continuous mapping $M : \mathbb{R}^n \rightarrow \mathbb{R}^n$ is a paracontraction, with respect to a norm $\|\cdot\|$ on $\mathbb{R}^n$, if $ \|M(x) - y\| < \|x - y\|,$ for all $x,y \in \mathbb{R}^n \text{ such that } x \notin \mathrm{fix}(M), y \in \mathrm{fix}(M)$.

\section{Background on Coalitional Games}\label{sec:Mathematical background }
Let us first provide a brief mathematical background on coalitional game theory and then describe two payoff distribution processes namely, the payoff allocation and bargaining.\\ 
In dynamic context, a coalitional game consists of a set of agents, indexed by $\mathcal{I} = \{1, \ldots, N\}$, who cooperate to achieve selfish interests. This cooperation at each time $k \in \mathbb{N}$ results in the generation of utility, as defined by a value function $v^{k}$.
 \begin{definition}[Coalitional game {(\cite{Nedic2013}, Sec. II-A)}]\label{def: instantaneous TU coalitional game} 
Let $\mathcal{I} = \{1, \ldots, N\}$ be a set of agents. For each time $k \in \mathbb{N},$ an instantaneous coalitional game is a pair $\mathcal{G}^k = (\mathcal{I}, v^{k})$ where $v^k: 2^\mathcal{I} \to \mathbb{R}$ is a value function that assigns a real value, $v^k(S)$, to each coalition $S \subseteq \mathcal{I}$. A dynamic coalitional game is a sequence of instantaneous games, i.e., $\mathcal{G} = (\mathcal{I}, (v^{k})_{k \in \mathbb{N}})$.$\hfill \square$
 \end{definition}
 
For an instantaneous game, an instantaneous value of a coalition has to be distributed among the member agents of the coalition so that each agent receives a certain payoff.
\begin{definition} [Payoff vector] \label{def: payoff} 
Let $\mathcal{I} = \{1, \ldots, N\}$ be a set of agents and $S \subseteq \mathcal{I}$ be a coalition in an instantaneous coalitional game $\mathcal{G}^{k}=(\mathcal{I}, v^{k}), k \in \mathbb{N}$. Then, for each $i \in S$, the element $x_i^k$ of a payoff vector $\boldsymbol{x}^k \in \mathbb{R}^{|S|}$ represents the share of agent $i$ of the value $v^k  (S)$.$\hfill \square$
\end{definition}
Within the game, we assume that each agent $i  \in \mathcal{I}$ acts rationally and efficiently. This means that the payoff vector, given in Definition \ref{def: payoff}, proposed by an agent must belong to its \textit{bounding set} as defined next.
\begin{definition}[Bounding set (\cite{Nedic2013}, Sec. II-B )]\label{def: bounding set}
For an instantaneous game $\mathcal{G}^{k}=(\mathcal{I}, v^{k}), k \in \mathbb{N}$, the set
\begin{equation} \label{eq: bounding set}
\begin{array}{ll}
\mathcal{X}_i^{k} :=    \bigg\{x \in \mathbb{R}^N \mid & \sum_{j \in \mathcal{I}} x_j = v^k (\mathcal{I}),\\
  &  \sum_{j \in S} x_j \geq v^k  (S), \forall S \subset \mathcal{I} \text{ s.t. } i \in S \bigg\}
\end{array}
\end{equation}
denotes the bounding set of an agent $i \in S$. $\hfill \square$
\end{definition}
Since an agent agrees only on a payoff vector in its bounding set, we can conclude that a mutually agreed payoff shall belong to the intersection of the bounding sets of all the agents. Interestingly, this intersection corresponds to the core, the solution concept that relates to the stability of a grand coalition, i.e., a coalition of all agents. The idea of stability, in this context, is based on the disinterest of each agent in defecting a grand coalition. {Let us formalize the concept of core for instantaneous coalitional games, as in Definition \ref{def: instantaneous TU coalitional game}}. 
\begin{definition} [Instantaneous core]\label{def: inst core}
The core $\mathcal{C}$ of an instantaneous coalitional game $\mathcal{G}^{k} = (\mathcal{I},v^{k}), k \in \mathbb{N}$, is the following set of payoff vectors:
\begin{equation} \label{core}
\begin{array}{lll}
  \mathcal{C}(v^{k}) &:=  \bigg\{ x \in \mathbb{R}^N \mid  \sum_{i \in \mathcal{I}} x_i = v^k (\mathcal{I}),\\
  &\qquad \qquad \qquad \;  \sum_{i \in S} x_i \geq v^k  (S), \forall S \subseteq \mathcal{I}  \bigg\}, \\
   & \; = \bigcap_{i=1}^N \mathcal{X}_i^{k},
\end{array}
\end{equation}
with $\mathcal{X}_i^{k}$ as in (\ref{eq: bounding set}), which is also the intersection of the individual bounding sets. $\hfill \square$
\end{definition}
Each payoff allocation that belongs to the core stabilizes the grand coalition, which implies that no agent or coalition $S \subset \mathcal{I}$ has an incentive to defect from the grand coalition.\\
In this paper, we consider a similar class of dynamic coalitional games { as in \cite{Nedic2013}}, where an instantaneous value of each coalition  $v^k(S)$ belongs to a finite set bounded by a minimum and a maximum value, i.e.,  $\underline{v}(S) \leq v^k(S) \leq \overline{v}(S)$. This restriction of values on $v^k$ gives rise to a family of games which we collectively regard as a robust coalitional game.
\begin{definition}[Robust coalitional game]\label{def: robust game}
 {Let $\mathcal{I} = \{1, \ldots, N\}$ index a set of agents. A robust coalitional game $\mathcal{R} = (\mathcal{I}, \mathcal{V})$, is a set of instantaneous coalitional games $(\mathcal{I}, v^{k})$ with $v^k  \in \mathcal{V}:= \{ u_1, u_2, \ldots, u_n \} \text{ with } |\mathcal{V}| < \infty $,  for all $k \in \mathbb{N}$, where each $u_l$ is a value function such that $\underline{u}(S) \leq  u_l(S) \leq \overline{u}(S)$ for all $S \subset \mathcal{I}$ and $u_l(\mathcal{I}) = \overline{u}(\mathcal{I})$. }
$\hfill \square$
\end{definition}
{In words, a robust coalitional game $(\mathcal{I}, \mathcal{V})$ is a family of a finite number of instantaneous coalitional games such that the value of the grand coalition $\mathcal{I}$ is fixed. This setup adequately addresses the practical scenario of negotiations where after the formation of the grand coalition its value becomes certain. However, to compute a core payoff in (\ref{core}) anticipated values of sub-coalitions are also required, which involves uncertainty. We note that our formulation of robust coalitional game is called ``robust game" in \cite{Nedic2013}. Next, we formalize the core of a robust coalitional game as the robust core.} 
\begin{definition}[Robust core]\label{def: robust core} For a robust coalitional game  $\mathcal{R} = (\mathcal{I}, \mathcal{V}),$  the robust core is the intersection of all the possible instantaneous core sets, i.e., 
\begin{equation}\label{eq: robust core}
  \textstyle \mathcal{C}_0 := \bigcap_{v \in \mathcal{V}}\mathcal{C}(v).  
\end{equation} 
\end{definition} $\hfill \square$
\begin{remark}\label{rem: particular robust game}
Let $\mathcal{R} = (\mathcal{I}, \mathcal{V})$ be a robust coalitional game. If there exists $\overline{v} \in \mathcal{V}$ such that for all $k \in \mathbb{N}$, $v^k (\mathcal{I}) = \overline{v}(\mathcal{I})$ and $v^k  (S) \leq \overline{v}(S) \text{ for any coalition } S \subset \mathcal{I}$, then, {$\mathcal{C}_0 = \mathcal{C}(\overline{v}) $ and thus}  $\mathcal{C}_0 \subseteq \mathcal{C}(v^k)\text{ for all } v^k \in \mathcal{V}$. Consequently, if $\mathcal{C}_0 \neq \varnothing$ then $\mathcal{C}(v^k) \neq \varnothing \text{ for all } k \in \mathbb{N}$. $\hfill \square$
\end{remark}
In the sequel, we deal with the grand coalition only, therefore, we use the core as the solution concept. We note from (\ref{core}) that the core $\mathcal{C}(v^{k})$ is closed and convex. Furthermore, the robust core $\mathcal{C}_0$ in (\ref{eq: robust core}) is assumed to be nonempty throughout the paper. {Nonemptiness implies that even under the variations in coalitional values, a mutually agreeable payoff exists.  }
\begin{assum}\label{asm: nonempty robust core}
{The robust core is non-empty, i.e., $\mathcal{C}_0 \neq \varnothing.$}  
\end{assum}$\hfill \square$\\
Next, we discuss a possible strategy for finding a payoff vector that belongs to core, $\mathcal{C}_0$ in (\ref{eq: robust core}) of a robust game $\mathcal{R} = (\mathcal{I}, \mathcal{V})$. Since centralized methods for finding a payoff vector $ \boldsymbol{x} \in \mathcal{C}_0$ do not capture realistic scenarios of interaction among autonomous selfish agents, we propose distributed methods that allow agents to autonomously reach a common agreement on a payoff distribution.\\
The two payoff distribution methods which we focus on are distributed payoff allocation and distributed bargaining. The former is an iterative procedure in which, at each step, an agent $i$ proposes a payoff distribution $ \boldsymbol{x}_i \in \mathbb{R}^N$ by averaging the proposals of neighboring agents and by introducing an innovation factor. This procedure aspires to eventually reach a mutually agreed payoff among participating agents.\\
In a bargaining process, to propose a payoff distribution $ \boldsymbol{x}_i \in \mathbb{R}^N$, an agent $i$, after averaging the proposals of all agents, makes it compliant to its own interest. Bargaining procedure also aspires to reach a mutually agreed payoff. \\
Thus, in both methods, the proposed payoff distributions $(\boldsymbol{x}_i)_{i \in \mathcal{I}}$ must eventually reach consensus.
\begin{definition} [Consensus set]\label{def: consensus set}
The consensus set $\mathcal{A} \subset \mathbb{R}^{N^{2}}$ is defined as:
\begin{equation}\label{eq: consensus}
\mathcal{A} := \{\mathrm{col}(\boldsymbol{x}_1, \ldots, \boldsymbol{x}_N) \in \mathbb{R}^{N^{2}} \mid \boldsymbol{x}_i = \boldsymbol{x}_j, \forall i,j \in \mathcal{I}\}. 
\end{equation} $\hfill \square$
\end{definition}
In the sequel, we consider the problem of iteratively computing a mutually agreed, payoff vector in the core, i.e., $\boldsymbol{x}^{k} \to  \bar{ \boldsymbol{x}} \in \mathcal{A} \cap \mathcal{C} ^N $. We address this problem via distributed algorithms under the payoff allocation and bargaining frameworks. {Both algorithms, starting from any initial payoff proposal $\boldsymbol{x}^{0}$, converge to some consensual payoff in the robust core in  (\ref{eq: robust core}). }
\vspace{-2.5mm}
\section{Distributed Payoff Allocation}\label{sec: payoff allocation}
In coalitional games, the agents cooperate because they foresee a higher individual payoff compared to non-cooperative actions. A payoff that can sustain such cooperation, referred as a stable payoff, shall satisfy the criteria in (\ref{core}). Thus, the goal of a payoff allocation process is to let the agents achieve a consensus on a stable payoff in a distributed manner. During the allocation process, each agent proposes a payoff for all the involved agents based on the previous proposals of his neighbors and an innovation term.\\
In this section, we propose a payoff allocation in the context of robust coalitional games, where the value function $v$, at each iteration $k$, takes a value within the given bounds. { We model our setup in  a distributed paradigm, where each agent estimates the coalitional values independently, hence during each iteration different agents can assign different values to the same coalition. In context of a robust coalitional game $(\mathcal{I}, \mathcal{V})$, this distributed evaluation of the coalitional values implies that at each negotiation step, an agent can independently choose any value function $v$ from a family  $\mathcal{V}$, without central coordination. We prove that even under the distributed evaluation of the value function by the agents, and the variation of the coalitional values, the proposed payoff allocation algorithm converges to a stable payoff distribution.} In particular, our goal is to construct a distributed fixed-point algorithm, using which the agents can reach consensus (\ref{eq: consensus}) on a payoff distribution that belongs to the robust core in (\ref{eq: robust core}). 
\subsection{Distributed payoff allocation algorithm}\label{subsec: Distributed Payoff Allocation Algorithm}
Consider a set of agents $\mathcal{I} = \{1, \ldots, N\}$ who synchronously propose a distribution of utility at each discrete time step $k \in \mathbb{N}$. Specifically, each agent $i \in \mathcal{I}$ proposes a payoff distribution ${ \boldsymbol{x}}_i^{k} \in \mathbb{R}^N$, where the $j$th element denotes the share of agent $j$ proposed by agent $i$ at iteration $k \in \mathbb{N}$.\\
Let the agents communicate over a time-varying network represented by a graph $G^{k}=(\mathcal{I}, \mathcal{E}^{k}) $, where $(j, i) \in \mathcal{E}^{k}$ means that there is an active link between the agents $i$ and $j$ at iteration $k$ and they are then referred as neighbours. Therefore, the set of neighbors of agent $i$ at iteration $k$ is defined as \(\mathcal{N}_{i}^k:=\left\{j \in \mathcal{I} |(i, j) \in \mathcal{E}^{k}\right\}\). We assume that at each iteration $k$ an agent $i$ observes only the proposals of its neighbouring agents. Furthermore, we assume that the union of the communication graphs over a time period of length $Q$ is connected. The following assumption is typical for many works in multi-agent coordination, e.g. \cite[Assumption 3.2]{nedic2017achieving}. 
\begin{assum}[$Q-$connected graph]\label{asm: Q-con }
For all $k \in \mathbb{N}$, the union graph $(\mathcal{I}, \cup_{l=1}^{Q} \mathcal{E}^{l+k})$  is strongly connected for some integer $Q \geq 1$.  $\hfill \square$
\end{assum}
The edges in the communication graph $G^{k}$ are weighted using an adjacency matrix $W^{k} = [w_{i,j}^k]$, whose element $w_{i,j}^k$ represents the weight assigned by agent $i$ to the payoff distribution proposed by agent $j$, ${ \boldsymbol{x}}_j^{k}$. Note that, for some $j$, $w_{i,j}^k = 0$ implies that $j \notin \mathcal{N}_{i}^k $ hence, the state of agent $i$ is independent from that of agent $j$. We assume the adjacency matrix to be doubly stochastic with positive diagonal, as assumed in \cite[Assumption 3.3]{nedic2017achieving}, \cite[Assumption 2, 3]{nedic2010constrained}. 
\begin{assum}[Stochastic adjacency matrix]\label{asm: graph}
 For all $k \geq 0$, the adjacency matrix $W^{k} = [w_{i,j}^k]$ of the communication graph $G^k$ satisfies following conditions:
    \begin{enumerate}
        \item It is doubly stochastic, i.e., \(\sum_{j=1}^{N} w_{i,j}=\sum_{i=1}^{N} w_{i,j}=1\);
        \item its diagonal elements are strictly positive, i.e., $w_{i,i}^k > 0, \forall i \in \mathcal{I}$;
        \item $\exists$ $\gamma > 0$ such that $w_{i,j}^k \geq \gamma$ whenever $w_{i,j}^k > 0$. $\hfill \square$
    \end{enumerate} 
\end{assum}
{Assumptions \ref{asm: Q-con } and \ref{asm: graph} ensure that the agents communicate sufficiently often to each other and have sufficient influence on the resulting allocation.} We further assume that the elements of communication matrix $W^{k}$ take values from a finite set hence, finitely many adjacency matrices are available.
\begin{assum}[Finitely many adjacency matrices]\label{asm: fixed graph}
The adjacency matrices $\{W^k\}_{k \in \mathbb{N}}$, of the communication graphs belong to $\mathcal{W}$, a finite family of matrices that satisfy Assumption \ref{asm: graph}, i.e., $W^k \in \mathcal{W}$ for all $k \in \mathbb{N}$.  $\hfill \square$
\end{assum}
This assumption on the adjacency matrices allows us to exploit important results from the literature regarding finite families of mappings for proving convergence of our algorithms.\\
In our setup, at iteration $k$, each agent $i$ proposes a payoff allocation $\boldsymbol{x}_i^{k+1}$, for all agents $j \in \mathcal{I}$, as a convex combination of its estimate $\boldsymbol{x}_i^k$ and an innovation term. To generate the latter, agent $i$ first takes an average of the observed estimates of its neighbors $\boldsymbol{x}_j^k, j \in \mathcal{N}_i^k$, weighted by an adjacency matrix, and then applies an operator $T_i^k$ on the evaluated average.\\
Specifically, we propose the following update rule for each agent $i \in \mathcal{I}$:
\begin{equation*}
  \boldsymbol{x}_i^{k+1}=(1-\alpha_k) \boldsymbol{x}_i^{k}+\alpha_k T_i^{k}\left(\textstyle \sum_{j=1}^{N} w_{i,j}^k \boldsymbol{x}_{j}^{k}\right), 
\end{equation*}
that is, in collective compact form,
\begin{equation}\label{main_it_allocation}
\boldsymbol{x}^{k+1} = (1-\alpha_k) \boldsymbol{x}^{k} + \alpha_k  \boldsymbol{T}^k \boldsymbol{W}^k(\boldsymbol{x}^{k}),
\end{equation}
where $(\alpha_k)_{k \in \mathbb{N}} \in [\epsilon, 1-\epsilon]$ for some $\epsilon \in (0,1/2]$, $\boldsymbol{T}^k (\boldsymbol{x}):= \mathrm{col}(T_1^{k}(\boldsymbol{x}_1), \ldots,  T_N^{k}(\boldsymbol{x}_N)) $ and $\boldsymbol{W}^k := W^{k} \otimes I_N $ represents an adjacency matrix.\\
In (\ref{main_it_allocation}), we require the operator $T_i^k$ to be nonexpansive and its fixed-point set to include the robust core in (\ref{eq: robust core}). For example, $T_i^k$ can be the projection onto the core, i.e.,  $ T_i^k = \mathrm{proj}_{\mathcal{C}(v^k)}$. 
 \begin{assum}[Nonexpansiveness]\label{asm: fixed points of T}
 For all $k \in \mathbb{N}$, the operator $\boldsymbol{T}^k$ in (\ref{main_it_allocation}) is such that $\boldsymbol{T}^k \in \mathcal{T}$, where $\mathcal{T}$ is a finite family of nonexpansive operators such that $\bigcap_{\boldsymbol{T} \in \mathcal{T}} \mathrm{fix}(\boldsymbol{T}) = \mathcal{C}^N_0$, with $\mathcal{C}_0$ being the robust core in (\ref{eq: robust core}). $\hfill \square$
 \end{assum}
Let us elaborate on this assumption in context of a robust coalitional game $ \mathcal{R} = (\mathcal{I}, \mathcal{V})$, as in Definition \ref{def: robust game}. Here, for all $k \in \mathbb{N}$, we assume that an instantaneous core $\mathcal{C} (v^{k})$ in (\ref{core}) generated by the value function $v^k \in \mathcal{V}$ is the fixed-point set of an operator $T_i^k \text{ for all }i\in \mathcal{I}$ which implies that $\mathrm{fix}(\boldsymbol{T}^k) = \mathcal{C}^N (v^k)$. Consequently, the intersection of the fixed-point sets of the operators $\boldsymbol{T}^k \in \mathcal{T}$ corresponds to the robust core in (\ref{eq: robust core}), i.e., $\bigcap_{\boldsymbol{T} \in \mathcal{T}} \mathrm{fix}(\boldsymbol{T}) = \bigcap_{v \in \mathcal{V}}\mathcal{C}^N (v) = \mathcal{C}^N_0 $.
Furthermore, we note that having a finite family of nonexpansive operators implies that the value function $v^{k}$ can only take finitely many values within a specified set. This limitation does not pose a significant hindrance in practical scenarios. First, because the number of discrete values inside bounded intervals can be arbitrarily large and secondly, because the most common interpretation of value is in a monetary sense, which is always rounded off to some currency division.\\
Next, we assume that each $\boldsymbol{T}^k \in \mathcal{T}$ appears at least once in every $Q$ iterations of (\ref{main_it_allocation}), with $Q$ being the integer in Assumption \ref{asm: Q-con }, which can be arbitrarily large.
\begin{assum}\label{asm: Q admissible} 
Let $Q$ be the integer in Assumption \ref{asm: Q-con }. The operators $(\boldsymbol{T}^k)_{k \in \mathbb{N}}$ in (\ref{main_it_allocation}) are such that, for all $n \in \mathbb{N}$, $\bigcup_{k=n}^{n+Q}\{\boldsymbol{T}^k\} = \mathcal{T}$, with $\mathcal{T}$ as in Assumption \ref{asm: fixed points of T}. $\hfill \square$
\end{assum}
{This assumption ensures that the resulting robust core in (\ref{eq: robust core}) correspond to all the value functions that belong to the family $ \mathcal{V}$. }
Under Assumptions \ref{asm: nonempty robust core}$-$\ref{asm: Q admissible}, we can guarantee the convergence of the state in iteration (\ref{main_it_allocation}) to {some} payoff in the set $\mathcal{A}\cap \mathcal{C}_0^N$, as formalized in the following statement.  
\begin{theorem}[Convergence of payoff allocation]\label{thm: payoff allocation}
Let Assumptions \ref{asm: nonempty robust core}$-$\ref{asm: Q admissible} hold and the step sizes satisfy $\alpha_k \in [\epsilon, 1-\epsilon] \text{ for all } k \in \mathbb{N}$, for some $\epsilon > 0$. Then, {starting from any $\boldsymbol{x}^0 \in \mathbb{R}^{N^2}$,} the sequence \((\boldsymbol{x}^{k})_{k=0}^{\infty}\) generated by the iteration in (\ref{main_it_allocation}) converges to {some} $\bar{\boldsymbol{x}} \in \mathcal{A}\cap \mathcal{C}_0^N$, with $\mathcal{A}$ as in (\ref{eq: consensus}) and $\mathcal{C}_0$ being the robust core (\ref{eq: robust core}). $\hfill \square$
\end{theorem}
\subsection{Convergence Analysis}
To prove the convergence of the payoff allocation process in (\ref{main_it_allocation}), we build upon a well-known result on time-varying nonexpansive mappings, presented by Browder in \cite{browder1967convergence}. To proceed, let us first define the notion of admissible sequence and then recall Browder's result.
\begin{definition}[Admissible sequence (\cite{browder1967convergence}, Def. 5)] \label{def: admissible sequence} A function  $j: \mathbb{N}_{>0} \to \mathcal{D} \subseteq \mathbb{N}_{>0}$ is said to be an admissible sequence of integers in \(\mathcal{D}\) if for each integer $r \in \mathcal{D}$, there exists $m(r) \in \mathbb{N}_{>0}$ such that the image under the function \(j\) of \(m(r)\) successive integers contains $r$, i.e., \(r \in \{j(n), j(n+1), \ldots, j(n+m(r))\}\), for all $n \in \mathrm{dom}(j)$. $\hfill \square$
\end{definition}
For example, every $p-$periodic sequence, i.e., $\{j^k\}_{k \in \mathbb{N}}$ where $j^{k+p} = j^k$, is admissible with $m(r) = p \text{ for all } r \in \mathrm{ran}(j)$ and a sequence $\{j^k = k\}_{k \in \mathbb{N}}$ is a non-admissible sequence. 
\begin{lem}[\cite{browder1967convergence}, Thm. 5]\label{lem: browder}
Let $(U_{r})_{r \in \mathcal{D}}, \mathcal{D} \subseteq \mathbb{N}_{>0}$, be a (finite or infinite) sequence of nonexpansive mappings such that $C=\bigcap_{r \in \mathcal{D}} \mathrm{fix}\left(U_{r}\right) \neq \varnothing$. Let \(\left(\alpha_{k}\right)_{k \in \mathbb{N}}\) be a sequence where $\alpha_k \in [\epsilon, 1-\epsilon]$ for some $\epsilon \in (0,1/2]$, and let \((j^{k})_{k \in \mathbb{N}}\) be an admissible sequence of integers in \(\mathcal{D} .\) Then, the sequence $(\boldsymbol{x}^{k})_{k \in \mathbb{N}_{>0}}$ generated by 
\begin{equation*}\label{eq: browder}
\boldsymbol{x}^{k+1}:=\left(1-\alpha_{k}\right) \boldsymbol{x}^{k} + \alpha_{k} U_{j^{k}}(\boldsymbol{x}^{k})
\end{equation*}
converges to some $\bar{\boldsymbol{x}} \in C $.$\hfill \square$
\end{lem}
Next, we recall some useful properties of nonexpansive and paracontraction operators. 
\begin{lem}[Doubly stochastic matrix (\cite{Fullmer2018}, Prop. 5 ), (\cite{grammatico2017proximal}, Prop. 3)]\label{lem: Doubly stochastic matrix}
If $W$ is a doubly stochastic matrix then, the linear operator defined by the matrix $W \otimes I_{n}$ is nonexpansive. Moreover, if the operator $(W \otimes I_{n}) (\cdot)$ satisfies Assumption \ref{asm: graph} then, it is also a paracontraction with respect to the mixed vector norm $\|\cdot\|_{2,2}$. $\hfill \square$
\end{lem}
The fixed-point sets of nonexpansive and paracontraction operators relate to their compositions as follows.
\begin{lem}[Composition of nonexpansive operators (\cite{bauschke2011convex}, Prop. 4.49)]\label{prop:Composition of nonexpansive operators }
Let $T_1, T_2: \mathbb{R}^n \rightarrow \mathbb{R}^n$ be nonexpansive operators with respect to the norm $\|\cdot\|$. Then, the composition $T_1 \circ T_2$ is also nonexpansive with respect to the norm $\|\cdot \|$. Moreover, if either \(T_{1}\) or \(T_{2}\) is a paracontraction and \(\operatorname{fix} (T_{1}) \cap \operatorname{fix} (T_{2}) \neq \varnothing\) then, \(\mathrm{fix} (T_{1} \circ T_{2})=\mathrm{fix} (T_{1}) \cap \mathrm{fix} (T_{2})\). $\hfill \square$
\end{lem}
\begin{lem}[Composition of paracontracting operators (\cite{Fullmer2018}, Prop. 1)]\label{prop:Composition of paracontracting operators }
Suppose $M_1, M_2 : \mathbb{R}^n \rightarrow \mathbb{R}^n$ are paracontractions with respect to same norm $\|\cdot\|$ and $\mathrm{fix}(M_1) \cap \mathrm{fix}(M_2) \neq \varnothing$. Then, the composition $M_1 \circ M_2$ is a paracontraction with respect to the norm $\|\cdot \|$ and $\mathrm{fix}(M_1 \circ M_2) = \mathrm{fix}(M_1) \cap \mathrm{fix}(M_2)$. $\hfill \square$
\end{lem}
The Lemmas provided above are convenient operator-theoretic tools that help us in keeping our proofs elegantly brief. Using these tools, let us prove the following Lemma which we exploit later in the proof of Theorem \ref{thm: payoff allocation}.

\begin{lem}\label{lem: payoff allocation fixed-points}
Let $\boldsymbol{T}_1, \ldots, \boldsymbol{T}_q$ be a set of nonexpansive operators with  $\bigcap_{r=1}^q \mathrm{fix}(\boldsymbol{T}_r) = C$. Let the composition of the adjacency matrices that satisfy Assumption \ref{asm: graph}, i.e., $W_q  W_{q-1} \cdots  W_1 $ represent a strongly connected graph. Let $\boldsymbol{W}_r := W_{r} \otimes I_N $. Then, $ \bigcap_{r=1}^q \mathrm{fix}(\boldsymbol{T}_r \boldsymbol{W}_r) = \mathcal{A} \cap C$, where $\mathcal{A}$ is the consensus set in (\ref{eq: consensus}). $\hfill \square$
\end{lem}
\begin{proof} 
By Lemmas \ref{lem: Doubly stochastic matrix} and \ref{prop:Composition of nonexpansive operators }, $\mathrm{fix}(\boldsymbol{T}_r\boldsymbol{W}_r) = \mathrm{fix}(\boldsymbol{T}_r) \cap \mathrm{fix}(\boldsymbol{W}_r)$ hence, $\bigcap_{r=1}^q \mathrm{fix}(\boldsymbol{T}_r \boldsymbol{W}_r) = \mathrm{fix}(\boldsymbol{T}_r) \cap \mathrm{fix}(\boldsymbol{W}_r) \cap \cdots \cap \mathrm{fix}(\boldsymbol{T}_1) \cap \mathrm{fix}(\boldsymbol{W}_1) $. By Lemmas \ref{lem: Doubly stochastic matrix} and \ref{prop:Composition of paracontracting operators }, $\bigcap_{r=1}^q \mathrm{fix}(\boldsymbol{W}_r) = \mathrm{fix}(\boldsymbol{W}_q \cdots \boldsymbol{W}_1)$ where, by the Perron-Frobenius theorem, $\mathrm{fix}(\boldsymbol{W}_q \cdots \boldsymbol{W}_1) = \mathcal{A}$. Since $ 
\bigcap_{r=1}^q \mathrm{fix}(\boldsymbol{T}_r) = C$,  we conclude that $\bigcap_{r=1}^q \mathrm{fix}(\boldsymbol{T}_r \boldsymbol{W}_r) = \mathcal{A} \cap C$.
\end{proof}
Given these results, we are now ready to prove Theorem \ref{thm: payoff allocation}.
\begin{proof}(Theorem \ref{thm: payoff allocation}).
Let us define the operator $ \boldsymbol{U}_f:= \boldsymbol{T}_f\boldsymbol{W}_f $ with $\boldsymbol{T}_f \in \mathcal{T}$ and $W_f \in \mathcal{W}$, where $\boldsymbol{W}_f := W_{f} \otimes I_N $. We note that, by Assumptions \ref{asm: fixed graph} and \ref{asm: fixed points of T} there are only finitely many such operators and therefore, we can define the operator family  $ \mathcal{U}: = \{\boldsymbol{U}_f\}_{f=1}^F$. Let $l:\mathcal{U} \to \mathbb{N}$ be a function such that $l(\boldsymbol{U}_f)$ gives the maximal length of the sequence which contains the operator $\boldsymbol{U}_f$. Furthermore, let $\mathcal{D} = \{f \mid l(\boldsymbol{U}_f) < \infty \} \subseteq \{1, \ldots F \}$, i.e., the set of indices of the operators that occur at least once in a finite length interval. Since $F$ is finite, there always exist an integer representing the length of sequences in which each index $f \in \mathcal{D}$ appears at least once, thereby fulfilling the admissibility condition in Definition \ref{def: admissible sequence}. Thus, by Lemmas \ref{lem: browder} and \ref{prop:Composition of nonexpansive operators }, the iteration in (\ref{main_it_allocation}) converges to some $\bar{\boldsymbol{x}} \in \bigcap_{f \in \mathcal{D}} \mathrm{fix}(\boldsymbol{U}_f)$. \\
Let $\mathcal{K}_L$ be the interval of a sequence containing $L$ consecutive operators from the family $\{\boldsymbol{U}_f\}_{f \in \mathcal{D}}$ such that $ \bigcap_{k \in \mathcal{K}_L} \mathrm{fix}(\boldsymbol{U}^{k}) = \bigcap_{f \in \mathcal{D}} \mathrm{fix}(\boldsymbol{U}_f)$. As we can choose an arbitrarily long interval, without loss of generality, let $L \geq Q$, with $Q$ being the integer in Assumptions \ref{asm: Q-con } and \ref{asm: Q admissible}. Then, it holds that $\bigcap_{k \in \mathcal{K}_L} \mathrm{fix}(\boldsymbol{U}^{k}) \subseteq \bigcap_{k \in \mathcal{K}_Q} \mathrm{fix}(\boldsymbol{U}^{k})$ because, having a longer interval of operators can either reduce the intersection set or leave it unchanged. Finally, by Lemma \ref{lem: payoff allocation fixed-points}, $\bigcap_{k \in \mathcal{K}_Q} \mathrm{fix}(\boldsymbol{U}^{k}) = \mathcal{A}\cap \mathcal{C}_0^N$.
\end{proof}
\begin{figure}[t]
\centering
\includegraphics[width=7 cm]{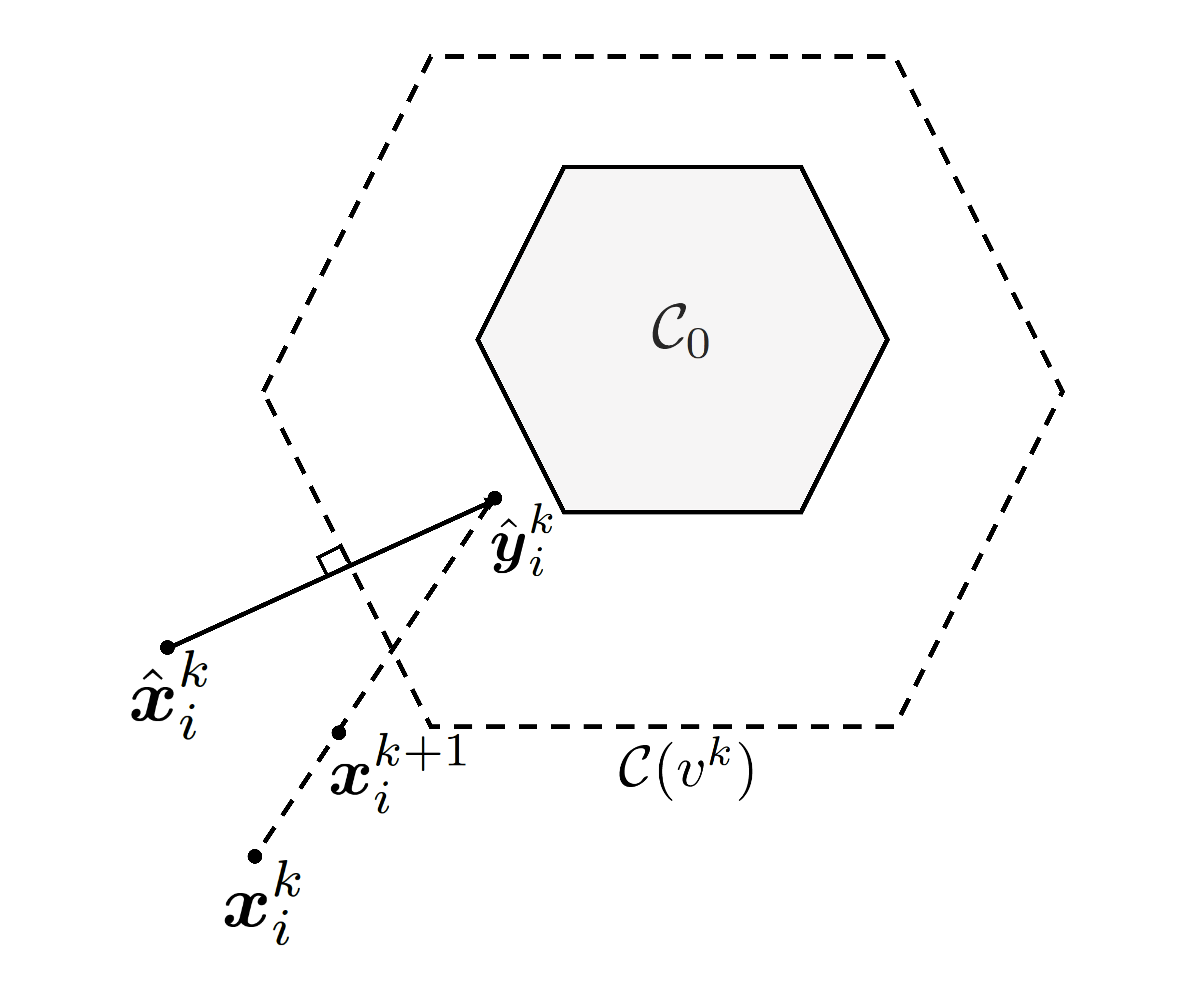}
\caption{Illustration of the payoff allocation proposed by an agent $i$, as in (\ref{eq: illustartion proposal}) where $\hat{\boldsymbol{y}}_i^k := \mathrm{overproj}_{\mathcal{C}(v^{k})} \hat{\boldsymbol{x}}_i^{k}$.}
\label{fig: illustartion proposal}
\vspace{-3.5mm}
\end{figure}
\vspace{-2mm}
\subsection{Discussion}\label{subsec: discussion payoff}
Let us now visualize a proposal of an agent $i$ in (\ref{main_it_allocation}) by employing an over-projection operator, i.e.,  $T_i^{k} = \mathrm{overproj}_{\mathcal{C}(v^{k})}$ which is a nonexpansive operator, see [\cite{bauschke2011convex}, Prop. 4.2].
For brevity, let $\hat{\boldsymbol{x}}_i^{k} := \sum_{j=1}^{N} w_{i,j}^k \boldsymbol{x}_{j}^{k}$. Then, the proposal of an agent $i$ reads as:
\begin{equation}\label{eq: illustartion proposal}
  \boldsymbol{x}_i^{k+1}=(1-\alpha_k) \boldsymbol{x}_i^{k}+\alpha_k \mathrm{overproj}_{\mathcal{C}(v^{k})} \hat{\boldsymbol{x}}_i^{k},
\end{equation}
where $\alpha_k \in [\epsilon, 1 - \epsilon] \text{ for some } \epsilon \in (0,1/2]$.\\ 
In Figure \ref{fig: illustartion proposal}, we illustrate an arbitrary instance of (\ref{eq: illustartion proposal}), where the proposed payoff allocation $\boldsymbol{x}_i^{k+1}$ does not belong to the instantaneous core $\mathcal{C}(v^{k})$ and hence it is not an acceptable payoff, even for agent $i$. Nevertheless, as stated in Theorem \ref{thm: payoff allocation}, repeated payoff allocations by all agents will eventually reach an agreement on the payoff that belongs to the robust core $\mathcal{C}_0$ in (\ref{eq: robust core}). Note that, in a payoff allocation process, intermediate allocation proposals can be irrational and therefore, the adoption of this process by a rational agent shall be motivated, e.g. via \textit{mechanism design}, where a central authority provides incentives to encourage cooperative behavior among agents and in turn drives the process towards the desired equilibrium.\\
{We remark that the number of possible coalitions grows exponentially in $N$, i.e., $2^N$, hence so does the computations required to evaluate the core by an individual agent. We also note that this feature is inherent in the class of coalitional games and in fact it is shared across the literature \cite{Nedic2013}, \cite{Bauso2015}. }
\section{Distributed bargaining protocol} \label{sec: bargaining}
{In this section, we propose a bargaining protocol under a typical negotiation framework and a distributed paradigm, similar to the payoff allocation in Section \ref{sec: payoff allocation}}. Specifically, at iteration $k$, each agent $i \in \mathcal{I}$ proposes a payoff distribution that belongs to its negotiation set, referred to as the bounding set $\mathcal{X}_i^k$ in  (\ref{eq: bounding set}). The intersection of negotiation sets represents the set of all plausible deals, i.e., the core and mutual agreement of agents on one such deal concludes the bargaining process. This struck deal corresponds to the final payoff distribution.
\subsection{Distributed bargaining algorithm}\label{subsec:Distributed Bargaining Algorithm }
For our distributed bargaining protocol, we use a similar setup as the payoff allocation algorithm (\ref{subsec: Distributed Payoff Allocation Algorithm}). Briefly, we consider a set of agents $\mathcal{I} = \{1, \ldots, N\}$, each of whom proposes a payoff distribution $ \boldsymbol{x}_i^{k} \in \mathbb{R}^{N}$ at each iteration $k \in \mathbb{N}$.   These agents communicate over a sequence of time-varying network graphs $(G^k)_{k \in \mathbb{N}}$, that satisfies Assumption \ref{asm: Q-con }, and the corresponding adjacency matrices $(W^{k})_{k \in \mathbb{N}}$ satisfy Assumptions \ref{asm: graph} and \ref{asm: fixed graph}.\\
During the negotiation, at each iteration $k$, an agent $i$ first takes an average of the estimates of neighboring agents $\boldsymbol{x}_j^k, j \in \mathcal{N}_i^k$, weighted by an adjacency matrix $W^k$, and then applies an operator $M_i^k$ on the resulting average. 
Specifically, we propose the following negotiation protocol for each agent $i \in \mathcal{I}$:
\begin{equation*}
  \boldsymbol{x}_i^{k+1}= M_i^{k}\textstyle \left(\sum_{j=1}^{N} w^k_{i,j} \boldsymbol{x}_{j}^{k}\right), 
\end{equation*}
that is, in collective compact form,
 \begin{equation}\label{main_it}
    { \boldsymbol{x}}^{k+1} = \boldsymbol{M}^k (\boldsymbol{W}^k { \boldsymbol{x}}^{k}),
\end{equation}
 where  $\boldsymbol{M}^k (\boldsymbol{x}):= \mathrm{col}(M_1^{k}(\boldsymbol{x}_1), \ldots,  M_N^{k}(\boldsymbol{x}_N))  $ and $\boldsymbol{W}^k := W^{k} \otimes I_N $ represents an adjacency matrix. \\
In (\ref{main_it}) we require the operator $M_i^k$ to be paracontraction, not necessarily a nonexpansive operator as in (\ref{main_it_allocation}). Utilizing a paracontraction operator allows us to prove convergence of our bargaining algorithm without the need of $\alpha-$averaging with the inertial term $\boldsymbol{x}^k$, as required for payoff allocation in (\ref{main_it_allocation}). Furthermore, in (\ref{main_it}), we also require the fixed-point set of $M_i^{k}$ to be the bounding set in (\ref{eq: bounding set}), i.e., $\mathrm{fix}(M_i^k) = \mathcal{X}_i^k$. Therefore, $ \mathrm{fix}(\boldsymbol{M}^k) = \bigcap_{i=1}^N \mathcal{X}_i^{k} = \mathcal{C}(v^{k}) $ and for a robust coalitional game $(\mathcal{I}, \mathcal{V})$, it holds that $ \bigcap_{v^k \in \mathcal{V}} \mathcal{C}(v^k) = \mathcal{C}_0$.
 \begin{assum}[Paracontractions]\label{asm: fixed points of M}
 For all $k \in \mathbb{N}$, $\boldsymbol{M}^k$ in (\ref{main_it}) is such that $\boldsymbol{M}^k \in \mathcal{M}$, where $\mathcal{M}$ is a finite family of paracontraction operators such that $\bigcap_{\boldsymbol{M} \in \mathcal{M}} \mathrm{fix}(\boldsymbol{M}) = \mathcal{C}^N_0$ with $\mathcal{C}_0$ being the robust core in (\ref{eq: robust core}). $\hfill \square$
 \end{assum}
Similar to the payoff allocation setup, we also assume that each $\boldsymbol{M}^k \in \mathcal{M}$ appears at least once in every $Q$ iterations of (\ref{main_it}), with $Q$ being the integer in Assumption \ref{asm: Q-con }.
\begin{assum}\label{asm: Q admissible bargaining} 
Let $Q$ be the integer in Assumption \ref{asm: Q-con }. The operators $(\boldsymbol{M}^k)_{k \in \mathbb{N}}$ in (\ref{main_it}) are such that, for all $n \in \mathbb{N}$, $\bigcup_{k=n}^{n+Q}\{\boldsymbol{M}^k\} = \mathcal{M}$, with $\mathcal{M}$ as in Assumption \ref{asm: fixed points of M}. $\hfill \square$
\end{assum}
Next, we formalize the main convergence result of the bargaining protocol in (\ref{main_it}).
\begin{theorem}[Convergence of bargaining protocol]\label{theorem: main}
Let Assumptions \ref{asm: nonempty robust core}$-$\ref{asm: fixed graph}, \ref{asm: fixed points of M}$-$\ref{asm: Q admissible bargaining} hold. Then, {starting from any $\boldsymbol{x}^0 \in \mathbb{R}^{N^2}$,} the sequence \((\boldsymbol{x}^{k})_{k=0}^{\infty}\) generated by the iteration in (\ref{main_it}) converges to {some} $\bar{\boldsymbol{x}} \in \mathcal{A}\cap \mathcal{C}_0^N$, with $\mathcal{A}$ as in (\ref{eq: consensus}) and $\mathcal{C}_0$ being the robust core (\ref{eq: robust core}). $\hfill \square$
\end{theorem}
\subsection{Convergence Analysis}
We prove the convergence of the bargaining protocol in (\ref{main_it}) by building upon a result related to the time-varying paracontractions, presented in \cite{Elsner1992}.
\begin{lem}[\cite{Elsner1992}, Thm. 1]\label{lemma: finite family}
Let $\mathcal{M}$ be a finite family of paracontractions such that $\bigcap_{M \in \mathcal{M}} \mathrm{fix}(M)$ $ \neq \varnothing $. Then, the sequence $(\boldsymbol{x}^k)_{k \in \mathbb{N}}$ generated by $\boldsymbol{x}^{k+1} := M^k( \boldsymbol{x}^{k})$ converges to a common fixed-point of the paracontractions that occur infinitely often in the sequence. $\hfill \square$
\end{lem}
In the following lemma, we provide a technical result about the composition of paracontractions which we exploit later in the proof of Theorem \ref{theorem: main}.

\begin{lem}\label{lemma: fix points}
Let $Q$ be the integer in Assumption \ref{asm: Q-con }. Let  $\boldsymbol{M}_1, \ldots, \boldsymbol{M}_Q $ be paracontraction operators with  $  
\bigcap_{r=1}^Q \mathrm{fix}(\boldsymbol{M}_r) =: C$ and let $W_Q  W_{Q-1} \cdots  W_1 $  be the composition of the adjacency matrices where $W_r \in \mathcal{W}$, with $\mathcal{W}$ as in Assumption \ref{asm: fixed graph}. Let $\boldsymbol{W}_r := W_r \otimes I_N $. Then, the composed mapping $ \boldsymbol{x} \mapsto (\boldsymbol{M}_Q\boldsymbol{W}_Q  \circ \cdots \circ  \boldsymbol{M}_1 \boldsymbol{W}_1)(\boldsymbol{x})$ 
\begin{enumerate}[(i)]
    \item is a paracontraction with respect to norm $\| \cdot \|_{2,2}$;
    \item $\mathrm{fix}(\boldsymbol{M}_Q\boldsymbol{W}_Q  \circ \cdots \circ  \boldsymbol{M}_1 \boldsymbol{W}_1) = \mathcal{A} \cap C$,  
\end{enumerate}
where $\mathcal{A}$ is the consensus set in (\ref{eq: consensus}). $\hfill \square$
\end{lem}
\begin{proof}
(i): It follows directly from Lemmas \ref{lem: Doubly stochastic matrix} and \ref{prop:Composition of paracontracting operators }.\\
(ii): By Lemmas \ref{lem: Doubly stochastic matrix} and \ref{prop:Composition of paracontracting operators }, $\mathrm{fix}(\boldsymbol{M}_Q\boldsymbol{W}_Q  \circ \cdots \circ  \boldsymbol{M}_1 \boldsymbol{W}_1) = \mathrm{fix}(\boldsymbol{M}_Q) \cap \cdots \cap \mathrm{fix}(\boldsymbol{M}_1) \cap \mathrm{fix}(\boldsymbol{W}_Q) \cap \cdots \cap \mathrm{fix}(\boldsymbol{W}_1)$. Again, by Lemmas \ref{lem: Doubly stochastic matrix} and \ref{prop:Composition of paracontracting operators }, $\bigcap_{r=1}^Q \mathrm{fix}(\boldsymbol{W}_r) = \mathrm{fix}(\boldsymbol{W}_Q \cdots \boldsymbol{W}_1)$ and since the composition $\boldsymbol{W}_Q \cdots \boldsymbol{W}_1$ is strongly connected, by the Perron-Frobenius theorem, $\mathrm{fix}(\boldsymbol{W}_Q \cdots \boldsymbol{W}_1) = \mathcal{A}$. Finally, as $\bigcap_{r=1}^Q \mathrm{fix}(\boldsymbol{M}_r) = C$, $\mathrm{fix}(\boldsymbol{M}_Q\boldsymbol{W}_Q  \circ \cdots \circ  \boldsymbol{M}_1 \boldsymbol{W}_1) = \mathcal{A} \cap C$. 
\end{proof}
Given these preliminary results, we are now ready to present the proof of Theorem \ref{theorem: main}. 
\begin{proof}(Theorem \ref{theorem: main})
Let us define the sub-sequence of $\boldsymbol{x}^{k} \text{ for all } k \in \mathbb{N}$ as $\boldsymbol{z}^t = \boldsymbol{x}^{(t-1)Q}$ for each $t \geq 2 $ with $Q$ being the integer in Assumptions \ref{asm: Q-con } and \ref{asm: Q admissible bargaining}. Then,
\begin{equation}\label{eq: z^{k} subsequence}
    \boldsymbol{z}^{t+1} = \boldsymbol{M}^{tQ - 1}\boldsymbol{W}^{tQ - 1}  \circ \cdots \circ  \boldsymbol{M}^{(t-1)Q} \boldsymbol{W}^{(t-1)Q} \boldsymbol{z}^t
\end{equation}
for $t \geq 2$. It follows from assertion 1 of Lemma \ref{lemma: fix points} that the maps $\boldsymbol{x} \longmapsto (\boldsymbol{M}^{tQ - 1}\boldsymbol{W}^{tQ - 1}  \circ \cdots \circ  \boldsymbol{M}^{(t-1)Q} \boldsymbol{W}^{(t-1)Q})(\boldsymbol{x}),$ $ t \geq 2 $ are all paracontractions. Also, under Assumption \ref{asm: fixed graph}, there can be only finitely many such maps. Furthermore, by assertion 2 of Lemma \ref{lemma: fix points}, the set of fixed-points of each map is
$\mathcal{A} \cap \mathcal{C} ^N$. Thus, by Lemma \ref{lemma: finite family}, the iteration in (\ref{eq: z^{k} subsequence}) converges to some $ \bar{\boldsymbol{z}} \in \mathcal{A} \cap \mathcal{C} ^N$. 
\end{proof}
\vspace{-2.5mm}
\subsection{Discussion}
In our proposed bargaining process in (\ref{main_it}), let $ \textstyle \boldsymbol{M}^k = \mathrm{proj}_{\mathcal{X}^{k}}$, for all $k \in \mathbb{N}$, which is a paracontraction [\cite{bauschke2011convex}, Prop. 4.16]. Then, the resulting iteration, i.e., $ \textstyle  \boldsymbol{x}^{k+1} = \mathrm{proj}_{\mathcal{X}^{k}} (\boldsymbol{W}^k { \boldsymbol{x}}^{k})$ reduces to the bargaining protocol presented in \cite{Nedic2013}. In that setup, the communication graphs and adjacency matrices also satisfy our Assumptions \ref{asm: Q-con } and \ref{asm: graph}, respectively. The bargaining algorithm in \cite{Nedic2013} lies within our bargaining framework, but with the exception that, in our setup, the value function $v^{k}$ can only take finitely many values in a bounded set. We emphasize that
our framework provides an agent with the flexibility to choose a paracontraction operator, not necessarily a projection. This allows an agent to propose a payoff on the boundary or in the interior of its bounding set.\\
{Finally, we note that in the bargaining process, the agents require a lower number of coalitional values to evaluate the bounding set compared to payoff allocation, i.e., $2^{N-1}$.}
\section{Numerical Simulations}\label{sec: illustrations}
 In this section, we present numerical illustrations of two realistic scenarios modeled as coalitional games with uncertain coalitional values. In the first scenario, we present a collaboration among three firms for providing abstract services; in the second scenario, we simulate the motivational application introduced in Section \ref{sec: intro}. Our goal for presenting the former is to illustrate the robust core and to differentiate between the structure of payoff allocation and bargaining processes during the negotiation stages. Therefore, we use only three agents (firms) to be able to illustrate the outcome graphically in dimension 2. Further, in the second simulation scenario, we demonstrate a more comprehensive application, namely cooperative energy storage optimization in a smart grid framework. 
 \vspace{-3.5mm}
\subsection{Illustrative example}
Consider three firms $\mathcal{I} = \{1, 2, 3\}$, which individually provide certain services to their customers. These firms can improve their efficiency by collaborating activities and hence generate a higher value. This collective value is a remuneration of services agreed upon by a customer and the coalition of firms in advance. To make this collaboration viable, all three firms have to agree upon their share of the generated value. The resulting scenario is a coalitional game among firms, a solution to which is an agreed payoff distribution in the core.\\
The core allocation in (\ref{core}) depends on the value of all possible sub-coalitions. In our example, the firms know with certainty about their individual values $v({\{i\}})$ and the value of the grand coalition $v(\mathcal{I})$, i.e., the final contract. However, the sub-coalitions are never formed and hence their values are unknown. We assume that the values of the sub-coalitions are random within a bounded interval. Under the above conditions, the coalitional game among the three firms takes the form of a robust coalitional game. Thus, we can apply the robust payoff distribution methods proposed in Sections \ref{sec: payoff allocation} and \ref{sec: bargaining}.\\
\begin{table}[]
 \centering

\caption{Coalitional values for the illustrative example}
\label{tab: coalitional values}
\begin{tabular}{ccccc} \toprule
${v({\{i\}})}, i \in \mathcal{I}$ & ${v({\{1,2\}})}$ & ${v({\{1,3\}})}$ & ${v({\{2,3\}}) }$ & $ {v({\{1,2,3\}})}$ \\ \midrule
 $1$ & $\{2,3,4\}$ & $\{2,3,4\}$ & $\{3,4,5\}$ & $8$ \\ \bottomrule
 \end{tabular}
 \vspace{-2.5mm}
 \end{table}
The coalitional values of this coalitional game among firms are given in Table \ref{tab: coalitional values}. For example, at each iteration $k$, the value function of the coalition $\{1,2\}$, i.e., $v({\{1,2\}})$, takes its value randomly from the set $\{2,3,4\}$ with uniform probability. The possibility of realizing only an integer value, with uniform probability, satisfies the assumption of finite operator families in Theorems \ref{thm: payoff allocation}$-$\ref{theorem: main} and also ensures that the resulting sequence satisfies Assumption \ref{asm: Q admissible}.  Furthermore, we consider a fixed, strongly connected communication graph which therefore satisfies Assumptions \ref{asm: Q-con }$-$\ref{asm: fixed graph}. For the initial proposals, we assume that each agent allocates entire value of coalition to itself, e.g. the initial proposal by firm $1$ will be $\boldsymbol{x}_1(1) = [\:8 \; 0 \; 0\:]^\top$.  Next, we evaluate the payoff distributions generated by payoff allocation algorithm in (\ref{main_it_allocation}) and the bargaining protocol in (\ref{main_it}).
\subsubsection{Distributed payoff allocation}
For implementation, we choose an over-projection operator, which is nonexpansive, and the step size $\alpha_k = 0.5$ for all $k \in \mathbb{N}$. The resulting iteration for each agent $i$ is as in (\ref{eq: illustartion proposal}).
  In Figure \ref{fig: allocation illustration}, we depict two arbitrary instances of the core set $\mathcal{C}^{'}, \mathcal{C}^{''}$ and the robust core $\mathcal{C}_0$ in (\ref{eq: robust core}). The allocation process in (\ref{main_it_allocation}) converges to consensus on the payoff allocation, $\bar{\boldsymbol{x}} = [2.4, 3, 2.6]$, which belongs to the robust core, i.e., $\mathcal{A}\cap \mathcal{C}_0^N$. An allocation in the robust core ensures that even under uncertainty on coalitional values, the collaboration will emerge as the only rational choice.
 We note that, in payoff allocation process each firm does not need to have deterministic information of the core, which is weaker from the usual assumption of coalitional games \cite{Bauso2015}. In fact, here the firms only  know the bounds on coalitional values.
 \setlength\belowcaptionskip{-2.5ex}
 \begin{figure}[t]
\centering
\includegraphics[width=7.6cm]{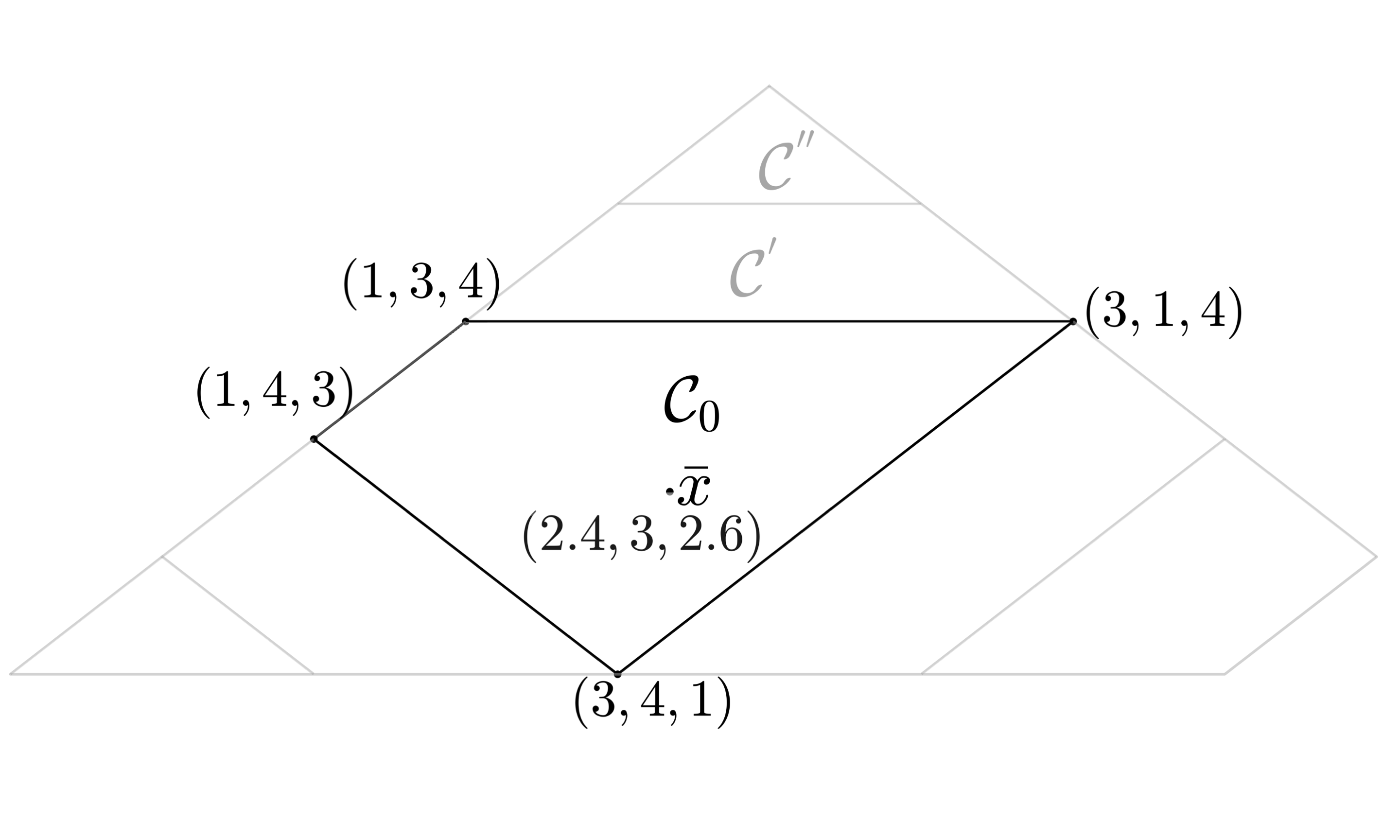}
\caption{Three instances of the core set, $\mathcal{C}_0, \mathcal{C}^{'},\mathcal{C}^{''}$ and final payoff allocation $\bar{\boldsymbol{x}}$.} 
\label{fig: allocation illustration}
\end{figure}
\subsubsection{Distributed bargaining protocol}
We implement the iteration in (\ref{main_it}), by using the projection operator, which is a paracontraction and therefore, it satisfies the assumptions of Theorem \ref{theorem: main}. In Figure \ref{fig: bargaining illustration}, we show an arbitrary negotiation step during the bargaining process. Here, a firm $i$ agrees with the payoff distribution only if it belongs to its bounding set $\mathcal{X}_i$. Thus, any mutually agreed payoff distribution must belong to the intersection of bounding sets, i.e., $\mathcal{C}=\bigcap_{i \in \mathcal{I}} \mathcal{X}_i$. Because of the uncertainty in the values of sub-coalitions, the bounding sets vary with iterations resulting in an instantaneous core as in (\ref{core}). The bargaining process in (\ref{main_it}) ensures convergence to the intersection of the instantaneous cores, i.e., the robust core $\mathcal{C}_0$ in (\ref{eq: robust core}). Thus, in our example, the resulting payoff distribution $\bar{\boldsymbol{x}}= [2.33, 2.833, 2.833]$ belongs to the set $\mathcal{A}\cap \mathcal{C}_0^N$.\\
Compared with the payoff allocation process, the knowledge requirement for the firms in the bargaining protocol is even weaker. Here, the firms are required to know the bounds on the values of their own sub-coalitions only, which is a reasonable assumption for a cooperation scenario. 
\vspace{-2mm}
\subsection{Cooperative energy storage optimization}
In this subsection, we simulate the cooperative ES optimization problem described in Section \ref{sec: intro} as a motivational example. We partially adapt the optimization setup from \cite{han2018constructing} and, additionally, introduce uncertainty in the RES generation.
\subsubsection{Problem setup}
Consider $N$ prosumers in an energy coalition $\mathcal{I}$, each equipped with RES generation and ES system. Our goal is to cooperatively optimize ES systems, by considering them as a single collective storage, for minimizing the coalitional cost in (\ref{eq: coalitional cost}) and, distribute the resulting cost savings, i.e., coalitional value in (\ref{eq: coalitional value}) among prosumers. Moreover, the share of each prosumer, i.e., the payoff should belong to the robust core in (\ref{eq: robust core}).
We compute the coalitional value of each coalition $S \subseteq \mathcal{I}$ for a time period of length $K$ by solving a linear optimization problem. We assume that the ES system of each prosumer $i$ has an energy capacity of \(e_{i} \geq 0,\) a charge and discharge limit, \(\overline{b}_{i} \geq 0\) and \(\underline{b}_{i} \geq 0\) respectively, a charge and discharge efficiency \(\eta_{i}^{\text{ch}}\) and \(\eta_{i}^{\text {dc }} \in(0,1)\), respectively. We also consider an initial state of charge for each ES, \(\text{SoC}_{i}^{0} \in[0,1]\) where $1$ represents a fully charged battery. We denote the amount of energy stored and released from agent $i$'s ES during time $t$ be $b_{i}^{t +}$ and  $b_{i}^{t -}$, respectively.\\
\setlength\belowcaptionskip{-2.5ex}
\begin{figure}[t]
\centering
\includegraphics[width=6.6cm]{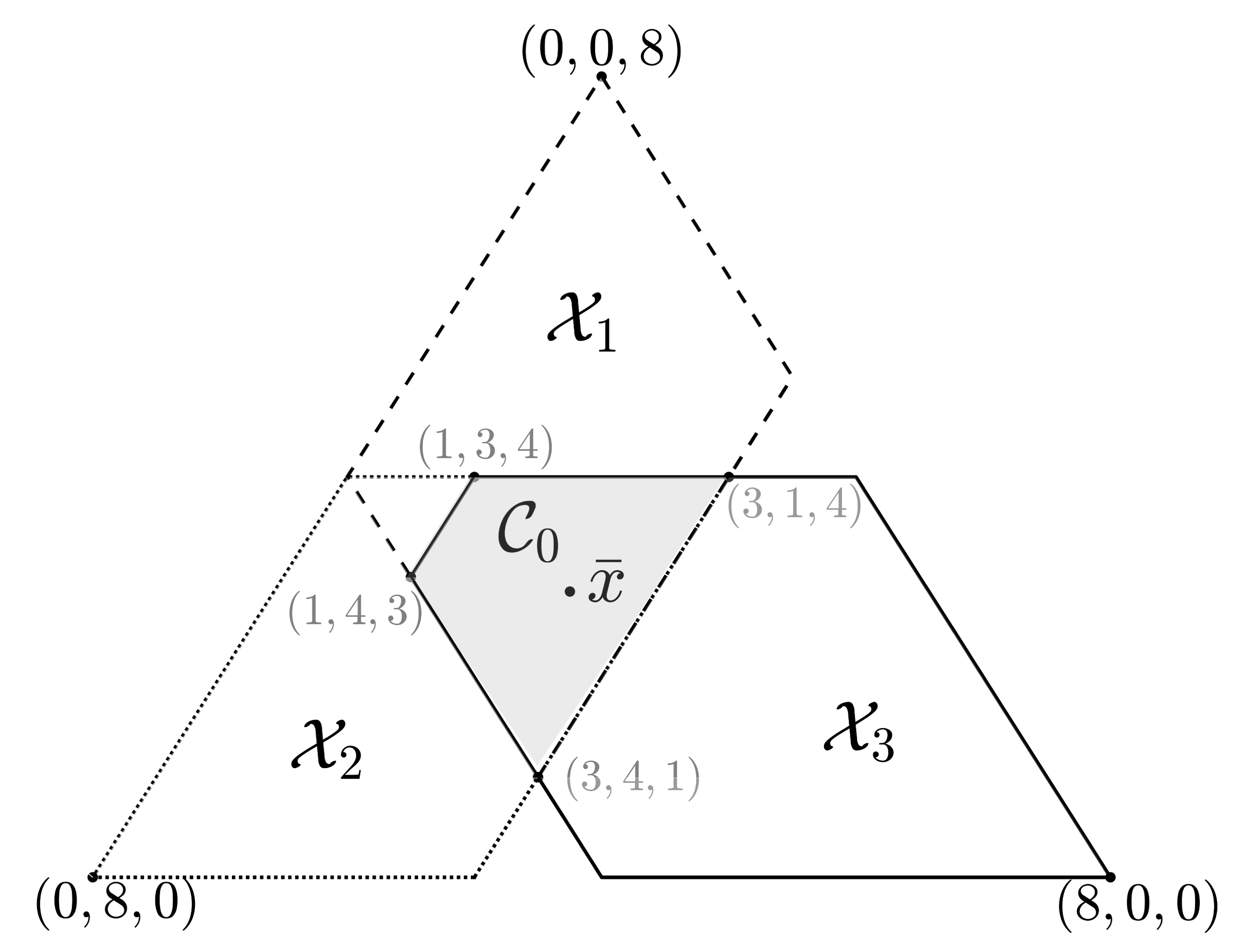}
\caption{An instance of bargaining process showing the bounding sets of the agents, $\mathcal{X}_1$, $\mathcal{X}_2$, $\mathcal{X}_3$, and the robust core $\mathcal{C}_0=\bigcap_{i \in \mathcal{I}} \mathcal{X}_i$. $\bar{\boldsymbol{x}}$ is the final payoff vector.} 
\label{fig: bargaining illustration}
\end{figure}
Next, let us denote the vectors representing charge and discharge energies of all prosumers by $\boldsymbol{b}^-$ and $\boldsymbol{b}^+$. Moreover, because of the difference in buying and selling prices of electricity, let us divide the coalitional net load into a positive part $\boldsymbol{L}^+$, which corresponds to the energy bought from the grid, and a non-positive part $\boldsymbol{L}^-$, which represents the energy sold to the grid. These four vectors are the decision variables of our ES optimization problem that computes the coalitional cost $c (S)$  for each coalition $S \subseteq \mathcal{I}$ as follows:
\begin{align}
\min _{\substack{\boldsymbol{b}^{+}, \; \boldsymbol{b}^{-},\\  \boldsymbol{L}^{+}, \; \boldsymbol{L}^{-}}} & \sum_{t=1}^{K}\bigg\{ p^{t}_{b}  \sum_{i \in S} L_{i}^{t +}+p^{t}_S  \sum_{i \in S} L_{i}^{t -}\bigg\}\tag{12a}\\ 
  \mathrm{s.t.} \quad &\: L_{i}^{t -} \leq 0 \leq L_{i}^{t +} \tag{12b}\\
& \sum_{i \in S} (b_{i}^{t +}+b_{i}^{t -}+q_{i}^{t}) \leq \sum_{i \in S} L_{i}^{t +} \tag{12c} \\
& \sum_{i \in S}(b_{i}^{t +}+b_{i}^{t -}+q_{i}^{t}) = \sum_{i \in S} (L_{i}^{t +}+L_{i}^{t -}) \tag{12d} \\
& \:\underline{b}_{i} \leq b_{i}^{t -} \leq 0 \leq b_{i}^{t +} \leq \overline{b}_{i} \tag{12e} \label{eq: charge limitation}\\
& \sum_{t=1}^{K}\left(b_{i}^{t +} \eta_{i}^{\text{ch}}+b_{i}^{t -} / \eta_{i}^\text{dc}\right)=0, \quad \forall i \in S \tag{12f} \label{eq: charge variation}\\
&\: 0 \leq e_{i} \text{SoC}_{i}^{0}+\sum_{t=1}^{m}\left(b_{i}^{t +} \eta_{i}^{\text{ch}}+b_{i}^{t -}/ \eta_{i}^\text{dc}\right) \leq e_{i} \tag{12g} \label{eq: initial charge}\\
&\: \forall i \in S, \forall t \in[1, K], \forall m \in[1, K] \nonumber.
\end{align} 
The constraints (\ref{eq: charge limitation})$-$(\ref{eq: initial charge}) are related to the physical limitations of ES systems. Specifically, (\ref{eq: charge limitation}) represents the limitation on the rate of charge/discharge, (\ref{eq: initial charge}) represents energy storage capacity and (\ref{eq: charge variation}) ensures that the state of charge of each ES at the end of the horizon $K$ is same as the initial, i.e.,  $\text{SoC}^K_i = \text{SoC}^0_i$. For further details, we refer to \cite{han2018constructing}. \\
To proceed, we introduce uncertainty in the net energy consumption $q_i^t$, since the generation of RES is uncertain. However,  $q_i^t$ can only realize values from the interval $[q_i^\text{min}, q_i^\text{max}]$ as explained in Section \ref{sec: intro}. Here, these bounds refer to the optimistic and conservative forecasts. 
\setlength\belowcaptionskip{-0.5ex}
\begin{figure}[t]
\begin{subfigure}{.5\textwidth}
  \centering
  \includegraphics[width=0.96\linewidth]{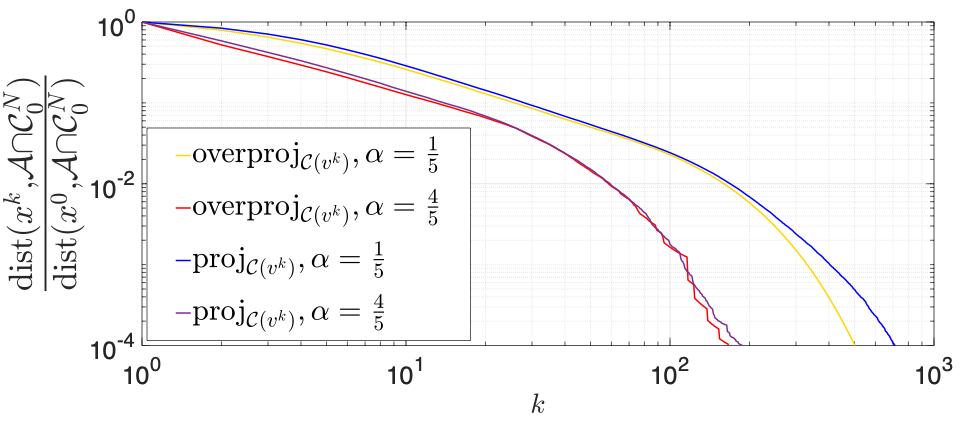}  
  \caption{}
  \label{fig: allocation_app_avg}
\end{subfigure}
\setlength\belowcaptionskip{-1ex}
\begin{subfigure}{.5\textwidth}
\centering
  \includegraphics[width=0.96\linewidth, height = 4.3cm]{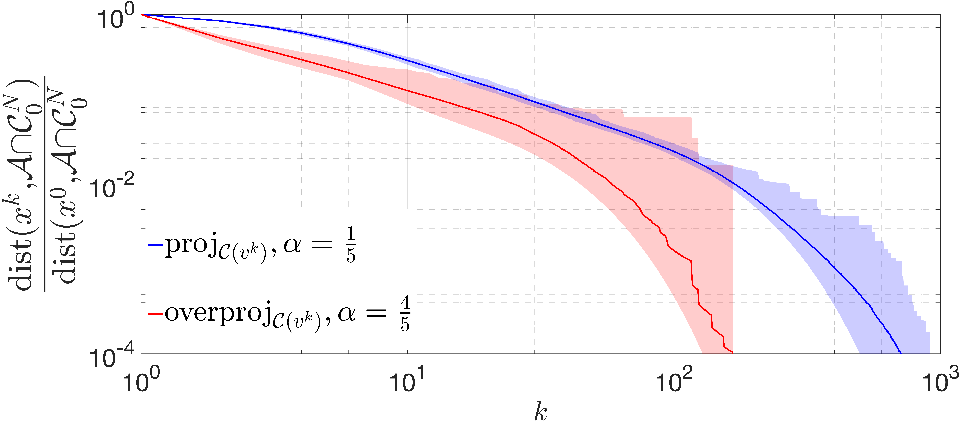}  
  \caption{}
  \label{fig: allocation_app_var}
\end{subfigure}
\setlength\belowcaptionskip{-2.5ex}
\caption{(a) Sampled average of the trajectories of $\mathrm{dist}(\boldsymbol{x}^k, \mathcal{A} \cap\mathcal{C}_0^N)/\mathrm{dist}(\boldsymbol{x}^0, \mathcal{A} \cap\mathcal{C}_0^N)$ for distributed allocation algorithm with operator $\mathrm{proj}_{\mathcal{C}{(v^k)}}$ for $ \alpha = 1/5, 4/5$ and $\mathrm{overproj}_{\mathcal{C}{(v^k)}}$ for $ \alpha = 1/5, 4/5$. (b) Sampled average of selected trajectories with spread of samples shown by shaded region.}
\vspace{-2.5mm}
\end{figure}
\subsubsection{ES optimization as a robust coalitional game}
Let us now put the optimisation setup, presented above, in the perspective of the payoff distribution problem. At the first stage, the grand prosumer coalition $\mathcal{I}$ optimizes their energy operation collectively via an aggregator over a time horizon of length $K$ and sells any expected excess of energy (available at each time interval $t$) to the retailer.  The coalition performs this process in advance and gets remunerated by the retailer. The additional value gained by the coalition as a result of the cooperation is given by (\ref{eq: coalitional value}). At the second stage, the attained coalitional value, i.e., $v(\mathcal{I})$ is distributed among the agents so that the payoff to each agent belongs to the robust core in (\ref{eq: robust core}). Thus, for the payoff distribution, an aggregator computes the value $v(S)$ for all $S \subset \mathcal{I}$ by solving the optimization problem presented above. To account for the uncertainty in the RES generation, the aggregator computes the bounds on the coalitional values as $\underline{v} (S) \leq v (S) \leq \overline{v}(S), S \subset \mathcal{I} $ and communicates the vector $v$ containing these bounds to all the agents, who in turn initiate the payoff distribution process. \\
This scenario, with uncertainty, requires robust solution and thus demonstrates practicality of our distributed allocation and bargaining algorithms. Furthermore, the core set is not singleton and different core payoffs can favour different agents. Therefore, the possibility of biased behavior of the aggregator can render a central computation of the payoffs unacceptable for prosumers. Thus, presented application further appreciates the distributed structure of the proposed algorithms.
\setlength\belowcaptionskip{-0.5ex}
\begin{figure}[t]
\begin{subfigure}{.5\textwidth}
  \centering
  \includegraphics[width=0.96\linewidth]{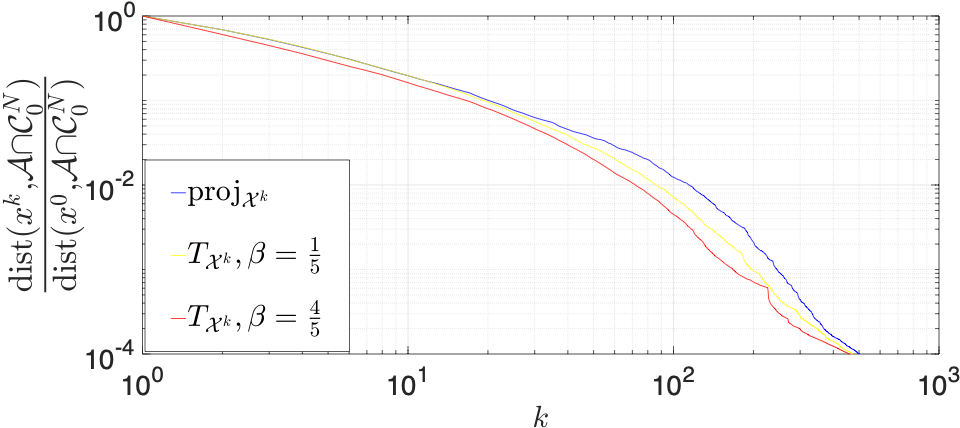}  
  \caption{}
  \label{fig: bargaining_app_avg}
\end{subfigure}
\setlength\belowcaptionskip{-1ex}
\begin{subfigure}{.5\textwidth}
 \centering
  \includegraphics[width=0.96\linewidth,  height = 4.3 cm]{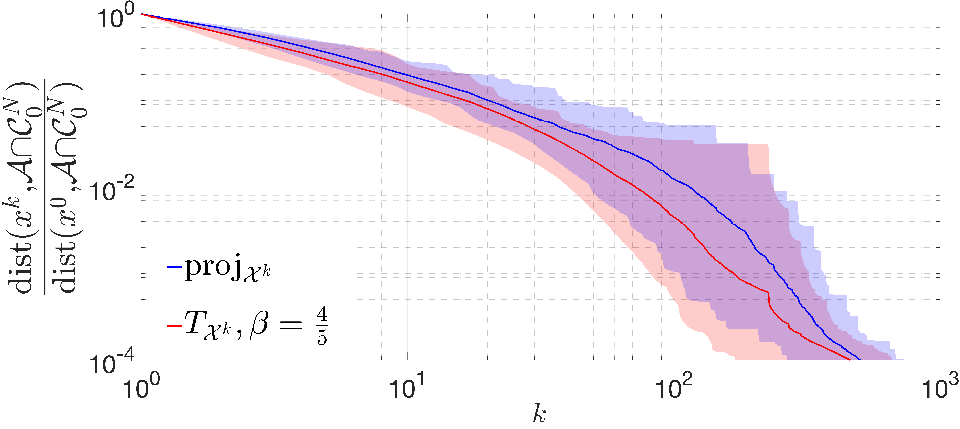}  
  \caption{}
  \label{fig: bargaining_app_var}
\end{subfigure}
\setlength\belowcaptionskip{-2.5ex}
\caption{(a) Sampled average of the trajectories of $\mathrm{dist}(\boldsymbol{x}^k, \mathcal{A} \cap\mathcal{C}_0^N)/\mathrm{dist}(\boldsymbol{x}^0, \mathcal{A} \cap\mathcal{C}_0^N)$ for distributed bargaining with operator $\mathrm{proj}_{\mathcal{X}^k}$ and $T_{\mathcal{X}^k}:= (1 - \beta) \mathrm{proj}_{\mathcal{X}^k}(\cdot) + \beta \mathrm{overproj}_{\mathcal{X}^k}(\cdot)$ for $ \beta = 1/5, 4/5$. (b) Sampled average of trajectories with spread of samples shown by shaded region.}
\end{figure}
\subsubsection{Simulations studies}
For the numerical simulation, we select a time horizon of $K =6$ hours {and an interval $t = 1$ hour. We consider a coalition of 6 prosumers} where each prosumer $i$ is equipped with the battery of energy capacity $e_i = 7$ kWh, a maximum charge power $\overline{b}_i = 3.5$ kW, a maximum discharge power $\underline{b}_i = 3.5$ kW, both charge and discharge efficiencies of $\eta_{i}^{\text{ch}} = \eta_{i}^\text{dc} = 95 \%$ and an initial state of charge $\text{SoC}_i^0 = 50\%$. We put the bounds of optimistic and conservative forecast on the RES generation of each agent and randomly generate net consumption scenarios. We then evaluate the coalitional value in (\ref{eq: coalitional value}) for each scenario and compute the bounds $\underline{v} (S) \text{ and } \overline{v} (S), S \subset \mathcal{I} $. We then run $100$ different trajectories of payoff distribution processes. We assume, for each trajectory, that agents initially allocate the whole value  $v(\mathcal{I})$ to themselves. {Also, to make sure that every prosumer's payoff proposal receives adequate importance and sufficient exposure, during negotiation, we assume a strongly connected communication graph among them which satisfies Assumptions \ref{asm: Q-con } and \ref{asm: graph}. Furthermore, as the coalitional value in cooperative energy optimization is in monetary terms, the prosumers consider reasonably rounded of units (dollars, cents etc.) which results in a finite set of points between the bounds $\underline{v} (S) \text{ and } \overline{v} (S)$, according to Assumption \ref{asm: Q admissible}.}\\
Moreover, for the distributed allocation process in (\ref{main_it_allocation}), the whole coalitional value vector $v$ is communicated to the agents whereas, for the bargaining process in (\ref{main_it}) only the value of agent's own coalitions are communicated. Finally, the agents initiate a robust coalitional game to reach the consensus on a payoff which belongs to the robust core in (\ref{eq: robust core}).  This payoff guarantees the stability of the grand coalition which in turn has considerable operational benefits for the power grid \cite{han2018incentivizing}.\\
We first report the numerical results for the distributed allocation process. In Figure \ref{fig: allocation_app_avg}, we compute the average of the sample trajectories obtained by $100$ runs and report the normalized distances $\mathrm{dist}(\boldsymbol{x}(k), \mathcal{C}_0\cap \mathcal{A})/\mathrm{dist}(\boldsymbol{x}(0), \mathcal{C}_0\cap \mathcal{A})$, for the projection and over-projection operators, by varying the parameter $\alpha$. We can observe that an over-projection operator with higher value of $\alpha$ results in faster convergence. In Figure \ref{fig: allocation_app_var}, we provide the spread of the sample trajectories to depict the best and worst convergence scenarios in our sample set.\\
Lastly, we simulate the distributed bargaining process in (\ref{main_it}) and report the average of the sample trajectories. In Figure \ref{fig: bargaining_app_avg}, we show the comparison of the normalized distances. We conduct the analysis by utilizing the projection operator and the convex combination of projection and over-projection operators, i.e., $T_{\mathcal{X}^k}:= (1 - \beta) \mathrm{proj}_{\mathcal{X}^k}(\cdot) + \beta \mathrm{overproj}_{\mathcal{X}^k}(\cdot)$ for varying $ \beta$. Both the operators are paracontraction operators \cite{bauschke2011convex}. Figure \ref{fig: bargaining_app_var} shows the spread of the sample trajectories. 
\vspace{-2mm}
\section{Conclusion}\label{sec: conclusion}
{We have addressed the problem of payoff distribution in robust coalitional games over time-varying communication networks. The goal is to make players reach a consensus on the payoff allocation that belongs to the robust core. 
We have shown that distributed payoff allocation and bargaining algorithms based on nonexpansive and paracontraction operators, e.g. over-projections, and network averaging converge consensually to the robust core, even with time-varying coalitional values.}
 \vspace{-2mm}


%





\ifCLASSOPTIONcaptionsoff
  \newpage
\fi





\bibliographystyle{IEEEtran}
\bibliography{IEEEabrv,Bibliography}

\vfill


\end{document}